\documentclass[11pt]{article}

\usepackage[margin=1in]{geometry}

\usepackage[T1]{fontenc}
\usepackage[utf8]{inputenc}

\usepackage{amsmath,amssymb,amsthm}

\usepackage{graphicx}%
\usepackage{multirow}%
\usepackage{mathrsfs}%
\usepackage[title]{appendix}%
\usepackage{xcolor}%
\usepackage{textcomp}%
\usepackage{manyfoot}%
\usepackage{booktabs}%
\usepackage{algorithm}%
\usepackage{algorithmicx}%
\usepackage{algpseudocode}%
\usepackage{listings}%

\usepackage[colorlinks=true,%
            linkcolor=blue,%
            citecolor=blue,%
            urlcolor=blue]{hyperref}

\def\bA{{\boldsymbol A}}
\def\bB{{\boldsymbol B}}
\def\bF{{\boldsymbol F}}
\def\bX{{\boldsymbol X}}
\def\bx{{\boldsymbol x}}
\def\by{{\boldsymbol y}}
\def\bw{{\boldsymbol w}}
\def\bp{{\boldsymbol p}}

\def\cK{{\mathcal K}}
\def\cL{{\mathcal L}}
\def\bbR{{\mathbb R}}

\newcommand{\bbet}{{\boldsymbol \beta}}
\newcommand{\bepsilon}{{\boldsymbol \epsilon}}
\newcommand{\blambda}{{\boldsymbol \lambda}}

\def\bxi{{\boldsymbol \xi}}
\def\btau{{\boldsymbol \tau}}

\newtheorem{theorem}{Theorem}

\newtheorem{remark}{Remark}

\title{An Entropic Factor Model for Robust Portfolio Replication}

\author{%
  Argimiro Arratia\thanks{Department of Computer Science, Universitat Polit\`ecnica de Catalunya, Jordi Girona s/n, Barcelona 08034, Spain. Email: \texttt{argimiro.arratia@upc.edu}. ORCID: 0000-0003-1551-420X}%
  \and 
  Henryk Gzyl\thanks{Center for Finance, IESA School of Business, Caracas 1010, Venezuela. Email: \texttt{henryk.gzyl@iesa.edu.ve}. ORCID: 0000-0002-3781-8848}%
}

\date{}

\begin{document}

\maketitle

\begin{abstract}
Portfolio replication, or the construction of a tradable basket of assets to match the risk-return 
profile of a target benchmark,  
is fundamentally an ill-posed inverse problem. 
When restricted to a subset of available assets, classical variance-minimizing models often yield
 unstable, over-leveraged portfolios highly vulnerable to market shocks. 
 %%%
We propose a unified, two-stage methodology rooted in information theory to achieve robust portfolio replication. First, we model the constituent asset returns against target factors, estimating parameters within data-driven empirical bounds via an entropy minimization principle. Second, using the same entropic approach, we determine the optimal weight replication. In both cases we use an entropy function of the Fermi-Dirac type defined directly on sets of constraints of the inverse problem. 
%%%
   We validate this Entropic Factor Model (EFM) against standard Ordinary Least Squares (OLS) across five numerical experiments, including standard equity tracking, multi-asset synthesis, and severe stress-test scenarios.
   Empirical results demonstrate that the EFM consistently outperforms OLS in terms of 
   annualized turnover and net-of-fees returns. Crucially, during the COVID-19 market crash and under severe idiosyncratic data corruption, the entropic framework acts as a probabilistic 
   ``circuit breaker", defensively reducing capital allocation to compromised assets 
   and providing a highly robust, risk-averse solution for generalized portfolio replication.
\end{abstract}




% Main text
\section{Introduction}

The ability to accurately replicate the financial characteristics of a target instrument using a given collection of available assets is a key task of quantitative finance. 
This generalized problem of portfolio replication extends far beyond the passive tracking of market-capitalization indices;
 it is the mathematical foundation for hedging complex liabilities, constructing synthetic multi-asset benchmarks, and engineering Factor Mimicking Portfolios (FMPs) 
 that allow investors to trade non-investable macroeconomic indicators. 
 
 Formally, defining a replicating portfolio amounts to assign weights $\bw = \{w_{1}, \dots, w_{N}\}$  
  to a given collection of assets with random returns $\{X_{1}, \dots, X_{N}\},$
 so that the portfolio return  $\sum_{n=1}^{N} w_{n}X_{n}$   matches the risk-return 
profile of the return of some target benchmark portfolio of  $K$ factors $F_1$, \ldots, $F_K$.

We decompose this problem into two stages:
\begin{enumerate}
    \item \textbf{Factor Extraction:} Determining the sensitivity (Beta) of each asset to the benchmark factors.
    \item \textbf{Portfolio Synthesis:} Finding the weights $\bw$ that match the benchmark's factors exposure.
\end{enumerate}

The first stage consists of an over-determined, ill-posed linear problem. Even though the matrix defining the system may be of full rank, the problem has additive errors in the data and box constraints upon the unknowns.
The second stage is also an ill-posed inverse problem, because the number of assets $N$ exceeds the number of factors $K$ ($N > K$), and there are box constraints upon the solution as well. 
 Standard approaches, such as minimizing the variance of the replication error via Ordinary Least Squares (OLS) regression or standard Quadratic Programming, often result in corner solutions or extreme weights due to the `error maximization' property of the optimizer \cite{michaud1989}, and its high sensitivity to input parameters \cite{best1991}.

To address these limitations, we propose a unified methodology, rooted in information theory, that applies a 
single entropy minimization procedure framework based on a Fermi-Dirac-like entropy function. 
We use a sign convention contrary to that used by physicists, and more in line with convex risk-minimization.

As a robust alternative to standard risk-minimization, the Principle of Maximum Entropy~\cite{jaynes}
offers a rigorous framework to process incomplete information without introducing arbitrary assumptions.  
 In the first stage, rather than fitting factor exposures via least squares, we adapt this entropic framework to estimate the parameters of a constrained factor model within data-driven, 
 empirical bounds.
 In the second stage, we apply the exact same entropic minimization to solve the inverse problem of determining the optimal replication weights. Operating as a unified framework, this entropic procedure acts as a natural penalization mechanism against epistemic uncertainty, ensuring that the optimizer does not assert false conviction in the presence of noisy or contradictory data.

We validate this Entropic Factor Model (EFM) through empirical analysis across five distinct numerical experiments, demonstrating its versatility as a generalized replication engine.
 First, a 
 standard equity index tracking scenario, consisting of partial replication of the S\&P 500, 
 the EFM achieves precision comparable to traditional methods but with reduced mean bias and superior directional stability. 
 Second, when tasked with synthesizing a complex, multi-asset universe incorporating both
  equities and cryptocurrencies, the model exhibits robust structural identification, selectively allocating weights based on pure factor exposure rather than indiscriminately smearing capital.
  Third, under a simulated ``flash crash'' stress-test, 
   the entropic framework demonstrates an inherent risk-aversion that standard models lack. Confronted with massive data corruption in a single asset, the model acts as a probabilistic 
   circuit breaker, recognizing the resulting volatility as structural uncertainty and drastically cutting exposure to the compromised asset.
   Fourth, partial replication under trading frictions, where a rolling-window backtest in a standard, benign market regime proves that the EFM achieves lower annualized turnover than OLS.
 Fifth and final,   a real market stress test based on replication of tradable SPY ETF, during the 2020 market crash produced by COVID-19. In this case  the EFM produced superior net returns and lower volatility compared to OLS.
  
The remainder of this paper is organized as follows. Section \ref{sec1} describes some related work and motivations.  Section  \ref{sec2} formally defines the portfolio replication problem and explains the two-stage entropy minimization methodology. 
%detailing the Fermi-Dirac-like formulation for portfolio weights.
Section \ref{sec3} outlines the mechanics of the empirical bounds used to constrain the factor loadings and the noise envelope. Section \ref{sec4} presents five  numerical experiments, providing a comparative analysis of the entropic model against standard least squares across standard replication, multi-asset synthesis, and severe stress-test regimes. Finally, Section \ref{sec5}  offers concluding remarks.

\section{Related Work and Motivation}\label{sec1}

The construction of a replicating portfolio
is a fundamental challenge in quantitative finance. 
While often discussed narrowly in the context of passive equity index tracking~\cite{roll92}, 
the scope of portfolio replication is vastly broader. 
It extends to the creation of Factor Mimicking Portfolios (FMPs) for exact arbitrage pricing~\cite{HKS87} 
and the engineering of economic tracking portfolios that allow investors to trade non-investable macroeconomic indicators, 
such as inflation or GDP growth~\cite{lamont}.

Because the replicating portfolio typically contains fewer assets than the broader market universe, and the constituent assets themselves exhibit high multicollinearity, 
determining the optimal capital allocation is mathematically an ill-posed inverse problem~\cite{illposed}. 
Standard variance-minimizing models, such as Ordinary Least Squares (OLS) or unconstrained Quadratic Programming, are highly sensitive to this multicollinearity. 
They often yield unstable, over-leveraged ``corner solutions'' that degrade significantly out-of-sample. 

To stabilize this inverse problem, the financial literature has heavily utilized norm-based regularization. Techniques such as $L_1$ (Lasso) and $L_2$ (Ridge) penalties have been 
widely applied to construct optimal sparse portfolios for generalized replication~\cite{FPW15} and specific high-dimensional index tracking \cite{BFP18}. 
More recently, the field has seen a surge in deep learning, graph-theoretic, and econophysics-based architectures aimed at capturing non-linear market characteristics and clustering asset structures. 
For instance, \cite{ZS24} demonstrated the use of one-dimensional Pointwise Convolutional Autoencoders combined with Shapley Additive Explanations (SHAP) to select assets and replicate indices during high-volatility regimes. 
In parallel, network topology and noise-filtering approaches have emerged to isolate robust asset interactions; Rubio-Garc{\'\i}a et al.~\cite{rubio2024} utilized graph filtering and network centrality metrics for sparse portfolio selection, while Grassetti~\cite{grassetti2025} combined Random Matrix Theory (RMT) filtering with Louvain community detection and eigenvalue centrality to mitigate correlation noise in index tracking.

Although these advanced neural architectures, graph-theoretic filters, and norm-based regularizers successfully mitigate unconstrained overfitting, enforcing realistic physical bounds or box constraints
 (e.g., position caps, leverage limits, or bounded noise domains) introduces distinct structural and operational challenges. 
In non-stationary return environments, bounded $L_1/L_2$ estimators rely on rolling cross-validation to tune penalty weights, exposing the model to hyperparameter instability and 
look-ahead noise. Furthermore, standard norm penalties shrink parameter estimates toward an uninformative zero baseline, ignoring the specific domain geometry defined by 
economic boundaries. From a computational point of view, solving box-constrained $L_p$ problems via 
active-set or coordinate-descent methods causes parameter estimates to jump abruptly onto and off hard boundaries 
under minor sampling noise. This boundary ``sloshing'' across rebalancing windows directly inflates 
portfolio turnover and execution friction. Simultaneously, deep neural architectures and multi-step heuristic pipelines remain 
computationally intensive, opaque ``black boxes'' whose continuous error-minimization objectives force the model to absorb extreme idiosyncratic outliers rather than bounding noise within explicit empirical envelopes.

As a robust, interpretable alternative, the Principle of Maximum Entropy, 
originally formulated in statistical mechanics by Jaynes~\cite{jaynes}, 
offers a rigorous framework to process incomplete information without introducing arbitrary assumptions. 
In quantitative finance, this principle has been leveraged to estimate risk-neutral probability distributions from option prices \cite{BK96} 
and to encourage natural portfolio diversification by avoiding extreme asset concentrations typical of mean-variance optimization \cite{AAHGSM24,BP08}. 

Our methodology builds upon this entropic tradition but uniquely extends it into a unified framework for the entire portfolio replication pipeline. 
Rather than utilizing entropy merely as an objective penalty for weight diversification, we treat both the structural factor loading estimation and the subsequent weight allocation 
as constrained inverse problems governed by information theory. 
This specific approach is a direct financial application of our recent mathematical work on solving ill-posed linear equations under strict constraints, in which we demonstrated that a Fermi-Dirac-like entropy function successfully regularizes ill-posed linear inverse problems with box constraints \cite{AAHG25}. 

By applying these entropic procedures, first to the robust estimation of asset factor exposures (Stage 1) and subsequently to the portfolio tracking weights (Stage 2), 
we provide a generalized replication engine. 
Unlike machine learning models that require vast amounts of training data to learn anomaly detection, the Entropic Factor Model (EFM) 
embeds risk-aversion directly into its physical bounds. 
Confronted with contradictory data, it acts as a probabilistic circuit breaker, inherently penalizing epistemic uncertainty during periods of severe data corruption.

\section{The entropic procedure}\label{sec2}

We propose a two-stage entropic framework to construct the replicating portfolio. Both stages are modeled as linear inverse problems with convex constraints of the form:
\begin{equation} \label{gen1}
    \bA \bxi = \by, \quad \bxi \in \mathcal{K}
\end{equation}
where $\bA$ is an $M \times N$ design matrix, $\by \in \mathbb{R}^M$ is the vector of constraints, and $\bxi \in \mathbb{R}^N$ is the vector of unknowns we wish to recover. 
The set $\mathcal{K} = \prod_{n=1}^N [a_n, b_n]$ defines the box constraints 
for each component of $\bxi$.
 Because these systems are typically underdetermined $(N>M)$, there are infinite solutions. 
 We select the unique solution $\bxi^*$ that minimizes the smooth convex function
 
\begin{equation}\label{obj}
\Psi(\bxi) = \sum_{j=1}^N\frac{\xi_j-a_j}{b_j - a_j}\ln\left(\frac{\xi_j-a_j}{b_j - a_j}\right) + \frac{b_j-\xi_j}{b_j - a_j}\ln\left(\frac{b_j-\xi_j}{b_j - a_j}\right)
\end{equation}

In other words, we restate  \eqref{gen1} as:

\begin{equation}\label{prob4}
\mbox{Find}\;\;\;\bxi^*=argmin\{\Psi(\bxi)| \bxi\in\cK;\,\bA\bxi=\by\}.
\end{equation}

The function $\Psi$ happens to be the Lagrange-Fenchel dual of the moment generating function  
(a Fermi-Dirac entropy function)
\begin{equation}\label{mgf}
\Phi(\btau) =  \sum_{j=1}^N \ln\bigg(e^{a_j\tau_j}+e^{b_j\tau_j}\bigg),\;\;\btau\in\bbR^{N},
\end{equation}
which is the logarithm of the Laplace transform of a measure that puts unit mass at each corner of $\cK.$ 
A simple computation yields:

\begin{gather}
\frac{\partial\Psi(\bxi)}{\partial \xi_j} = \frac{1}{b_j-a_j}\ln\bigg(\frac{\xi_j-a_j}{b_j-\xi_j}\bigg),\label{der1.1}\\
\frac{\partial \Phi(\btau)}{\partial \tau_j} = \frac{a_je^{\tau_j a_j}+b_je^{\tau_j b_j}}
{e^{\tau_j a_j}+e^{\tau_j b_j}}. \label{der1.2}
\end{gather}

With this, it is easy to establish that the functions $\Psi$ and $\Phi$ are related by:

\begin{equation}\label{FD1}
\begin{aligned}
\Psi(\bxi) = \sup\left\{\langle\bxi,\btau\rangle - \Phi(\btau)| \btau\in\bbR^N\right\}, \;\;\;\bxi\in\cK.\\
\Phi(\btau) = \sup\left\{ \langle\btau,\bxi\rangle - \Psi(\bxi) | \bxi\in\cK\right\},\;\;\;\btau\in\bbR^{N}.
\end{aligned}
\end{equation}

Both $\Phi(\btau)$ and $\Psi(\bxi)$ are convex, $\Psi$ is infinitely differentiable in the interior  of $\cK$, $int(\cK)$, 
and for $\btau\in\bbR^{N}$ the equation $\btau = \nabla\Psi(\bx)$ has a unique solution $\bx$ in the interior of $\cK$. Furthermore, the following identity holds:

\begin{equation}\label{FD2}
\nabla_\btau \Phi(\btau) = (\nabla_\bx\Psi)^{-1}(\btau),\;\;\;\mbox{whenever}\;\;\;\nabla_\bx\Psi(\bx) = \btau.
\end{equation}
That is, the gradients are inverse functions of each other. 
%%The interested reader may consider \cite{BL} for details about convexity, duality and all that. 

%which is given as:
%
%\begin{equation}\label{FD1}
%\Psi(\xi) = \sup_{\tau}\big(\langle\xi,\tau\rangle - \Phi(\tau)\big).
%\end{equation}
%
%Thus, 

\medskip
Now, to solve \eqref{prob4}, form the Lagrangian:
\begin{equation}\label{lagr}
\cL(\bxi, \blambda) = \Psi(\bxi)-\langle\lambda,(\bA\bxi-\by)\rangle,
\end{equation}
and differentiate with respect to $\bxi$ and $\blambda$ to obtain:

\begin{gather}
\nabla_{\bxi}\Psi(\bxi^*)-\bA^t\blambda^*=0,\;\;\label{sol1.1}\\
\bA\bxi^*-\by=0\;\;\;\label{sol1.2}
\end{gather}

Taking \eqref{FD2} into account we obtain:
\begin{equation}\label{repsol0}
\bxi^*=(\nabla_{\bxi}\Psi)^{-1}(\bA^t\blambda^*)=\nabla_{\btau}\Phi(\bA^t\blambda^*)
\end{equation}

This chain of equations prove the following theorem: 

\begin{theorem}\label{main}
Let $\Psi(\bx)$ and $\Phi(\btau)$ be related to each other as in \eqref{FD1}. Suppose that $\bA^t\blambda\in int(\cK)$
for any $\blambda\in\bbR^M.$  Then the solution to \eqref{prob4} is given by:

\begin{equation}\label{repsol1}
\xi_j^*= \frac{a_je^{a_j(\bA^t\blambda^*)_j} + b_je^{b_j(\bA^t\blambda^*)_j}}{e^{a_j(\bA^t\blambda^*)_j}+e^{b_j(\bA^t\blambda^*)_j}},\;\;\;\;\;\;j = 1, \ldots, N.
\end{equation}

Here $\blambda^* \in\bbR^{M}$ is the point at which $\Sigma(\blambda,\by) \equiv \langle\blambda,\by\rangle-\Phi(\bA\blambda)$ achieves its maximum value. 
$(\bA^t\blambda^*)_j$ is the $j$-th component of the vector $\bA^t\blambda^*$.
Also
\begin{equation}\label{dual1}
\Psi(\bxi^*) = \Sigma(\blambda^*,\by).
\end{equation}
\end{theorem}
\begin{proof}
From \eqref{repsol0} and \eqref{der1.2} it is clear that \eqref{repsol1} holds, 
and substituting in \eqref{sol1.2} it is clear that
\begin{equation}\label{inter1}
\bA\bxi^*-\by=\bA\nabla_{\btau}\Phi(\bA^t\blambda^*)-\by =\nabla_{\blambda}\bigg(\Phi(\bA^t\blambda)-\langle\by,\blambda\rangle)\bigg)|_{\lambda^*}=0
\end{equation}

\end{proof}

 The optimal solution  $\blambda^*$ can be obtained  
 %minimizing the dual potential function 
%$\Sigma(\lambda, y) := \langle\lambda,y\rangle-M(A^t\lambda)$
 using convex optimization techniques (e.g., Newton-Raphson or L-BFGS methods), as the dual problem is unconstrained and strictly convex~\cite{cvx}.
Most numerical software packages~\footnote{See \url{https://metrumresearchgroup.github.io/bbr/} for example, which combines the usual gradient method with a step reduction procedure at each iteration. This is convenient because the objective function may be very flat near the minimum} are written to solve a minimization problem by default, thus instead of maximizing $\Sigma(\blambda,\by),$ it is convenient to minimize $-\Sigma(\blambda,\by).$

Since conjugate gradient methods make use of the gradient of $\Sigma,$ it may be useful to have the computation at hand. The gradient of $\Sigma$ is
\begin{equation}\label{grad}
\frac{\partial \Sigma}{\partial \lambda_i} = \sum_{j=1}^N A_{i,j}\frac{a_je^{a_j(\bA^t\blambda)_j} + b_je^{b_j(\bA^t\blambda)_j}}{e^{a_j(\bA^t\blambda)_j}+e^{b_j(\bA^t\blambda)_j}} -y_i,\;\;\;\mbox{for}\;\;\;i=1, \ldots,M.
\end{equation}

\subsection{The reconstruction error}\label{recerr}

 When one solves numerically, the solution $\bxi^*$ need not satisfy $\bA\bxi^*=\by$ exactly. The reconstruction error just measures how large is the offset with respect to the problem data. In methods that use $\|\bA\bxi-\by\|^2$ as the objective function,  the value of the objective function at the optimum, namely $\|\bA\bxi^*-\by\|^2,$ is at the same time a measure of the reconstruction error. In our approach the minimum value $\Psi(\bxi^*)$ does not measure the quality of the reconstruction error. Nevertheless, we know from Theorem \ref{main} that:

\begin{equation}\label{error}
\|\nabla_{\blambda}\Sigma(\blambda,\by)\| = \|\bA\bxi^*(\blambda^*)-\by\|,
\end{equation}
 reaches its smallest value at $\blambda^*.$  The left hand side of \eqref{error} is the gradient of $\Sigma.$ If this value is zero, it means that the constraint is satisfied exactly. Actually, a preassigned size of this error is used as a halting criterion for the iterative minimization procedure: Once the norm of \eqref{error} is smaller than a preassigned amount (we use $10^{-5}$) the algorithm stops.

\section{The entropic solution to the portfolio replication problem}\label{sec3}
\subsection{Stage 1:  Factor Model Estimation}

Given the vector $\bX=(X_1, \ldots, X_N)$ of asset returns and the vector of factor returns 
$\bF=(F_1, \ldots, F_K)$, we posit the problem of modeling $\bX$ by the factors $\bF$ as follows:

\begin{gather}
\mbox{Find} \;\;\bbet_0\in[L_0,U_0]^N,\;\;\bepsilon\in\prod_{n=1}^N[-d_n,d_n],\;\mbox{and}\;\;\beta_{i,j}\in[L_{i,j},U_{i,j}]\;\; (1\le i\le N,\; 1\le j \le K)\;\;\nonumber\\ %\label{prob1.1}\\
\mbox{such that}\;\;\;\bX=\beta_0 + \bB\bF + \bepsilon.\;\;\;\label{prob1.2}  
\end{gather}
where $\bB = [\beta_{i,j}]_{(i,j) \in N\times K}$. 
The interpretation is standard. 
Denote by the usual $E[V]$ the expected value of a random variable $V$. 
 Take the expected value of both sides of \eqref{prob1.2}:
$$E[\bX]=\bbet_0+\bB E[\bF].$$
To interpret the coefficients and the entries of $\bB,$ note that
\begin{equation}\label{ders}
\frac{\partial E[X_i]}{\partial E[F_j]} = \beta_{i,j}, \;\;i=1, \ldots, N;\;\;\; j=1, \ldots, K.
\end{equation}
Therefore $\bbet_{0,i}$ is the part of the mean $E[X_i]$ not attributable to the factors. 
And this is why we can choose the means of the idiosyncratic factor $E[\bepsilon]=0.$ 
The $\bepsilon_i$ captures the variability of  $X_i$ not attributable to the variability of the factors.
 Each element $\beta_{i,j}$ (i.e., the $i$-th asset beta coefficient relative to the $j$-th factor, or factor loading) capture the influence, measured by $\partial E[X_i]/\partial E[F_j]$, of the factors. We  use this interpretation to establish the values of box constraints upon these coefficients.

The usual procedure to solve  %\eqref{prob1.1}-
\eqref{prob1.2} is to use a sample sequence of $T$ returns
for each of the variables, $X_i$, $i=1, \ldots, N$,  and for $F_j$, $j=1, \ldots, K$, 
and  determine $\beta_{0,i}$, $\beta_{i,j}$ and $\epsilon_i(t)$, such that 
$\beta_{0,i}\in[L_0,U_0],$ $\beta_{i,j}\in[L_{i,j},U_{i,j}]$ and $\epsilon_i(t) \in[-E_i,E_i]$ which satisfy:
\begin{equation}\label{E1}
X_i(t) = \beta_{0,i} + \sum_{j}\beta_{i,j} F_j(t) + \epsilon_i(t); \;\;\;t =1, \ldots,T.
\end{equation}
Now we think of $\bF$ as a $T\times K$ matrix, each column a factor $F_j$ and each row the $t$-th sample $F_j(t)$
of the factors.
Notice that now \eqref{prob1.2} has become a family of problems that can be solved in parallel.

To frame this as an entropic inverse problem ($\bA\bxi = \by$), we treat  the coefficients $\beta_{0,i}$, 
$\beta_{i,1}, \ldots, \beta_{i,K}$ and $\epsilon_i(t)$ ($t= 1, \ldots, T$) as unknowns to be determined simultaneously.
\begin{itemize}
    \item \textbf{The unknowns ($\bxi$):} The vector $\bxi$ is the concatenation of the coefficients and the error terms:
    \begin{equation}
        \bxi = [\beta_{0,i}\, \beta_i \, \epsilon_i]^t  \in \mathbb{R}^{1+K + T}
    \end{equation}
    where $\beta_i = (\beta_{i,1}, \ldots, \beta_{i,K})$ is the $i$-th row of the matrix $\bB$, and $\epsilon_i = (\epsilon_i(1), \ldots, \epsilon_i(T))$.
    \item \textbf{The design matrix ($\bA$):} The matrix $\bA$ is constructed by concatenating a column of 1 of length $T$, and augmenting the factor matrix ${\bF}$ with an identity matrix $I_T$ of size $T \times T$:
    \begin{equation}\label{eqfin}
        \bA = \begin{bmatrix} 1 & {\bF} & I_T \end{bmatrix} \in \mathbb{R}^{T \times (1+K + T)}
    \end{equation}
    \item \textbf{The target ($\by$):} The target vector is simply the observed asset returns of $X_i$.
\end{itemize}

The system becomes $[1 \,\, {\bF} \,\, I_T] \begin{bmatrix} \beta_{0,i} \\ \beta_i \\ \epsilon_i \end{bmatrix} = X_i$, subject to a convex set of constraints 
 $\mathcal{K}$ which defines the search space, and that we establish as follows.

We impose bounds on $\beta$ based on the empirical volatility ratio of the asset and factors, and bounds on $\epsilon$ using an envelope method based on the maximum observed residual from a robust central estimate. This prevents the solution from degenerating to OLS and ensures robustness to outliers.
The details are as follows.

To form the 
box constraining solutions for   beta coefficient $\beta_{i,j}$,  the lower and upper bounds are defined as
\begin{equation}\label{box1}
\begin{aligned}
L_{i,j} = \inf\{\frac{X_i(t)-X_i(t-1)}{F_j(t)-F_j(t-1)}|t=2,\ldots, T\}\\
U_{i,j} = \sup\{\frac{X_i(t)-X_i(t-1)}{F_j(t)-F_j(t-1)}|t=2, \ldots, T\}\\
\end{aligned}
\end{equation}

We use these quantities to obtain the constraints for the intercept as follows:

\begin{equation}\label{box2}
\begin{aligned}
L_{0,i} = \inf_{t,j}\{X_i(t)-U_{i,j}F_j(t)|j=1, \ldots,K;\,t=2,\ldots, T\}\\
U_{0,i} = \sup_{t,j}\{X_i(t)-L_{i,j}F_j(t)|j=1, \ldots,K;\,t=2,\ldots, T\}
\end{aligned}
\end{equation}

To specify the constraints for the $i$-th asset noise vector $\epsilon_i = (\epsilon_i(t): t=1, \ldots,T)$, 
 we employ a data-driven ``envelope'' method. 
 The goal is to ensure that the feasible region for the error term is wide enough to capture empirical deviations without allowing the solver to collapse to zero noise.

First, we define a ``central'' model by taking the midpoint of the feasible bounds for both the intercept and the beta coefficients:
\begin{equation}
\begin{aligned}
\hat{\beta}_{0,i}^{mid} &= \frac{1}{2} (L_{0,i} + U_{0,i}) \\
\hat{\beta}_{i,j}^{mid} &= \frac{1}{2} (L_{i,j} + U_{i,j}) \quad \text{for } j=1, \ldots, K
\end{aligned}
\end{equation}

Next, we calculate the residuals of the asset returns relative to this central model at each time $t$:
\begin{equation}
\hat{r}_i(t) = X_i(t) - \left( \hat{\beta}_{0,i}^{mid} + \sum_{j=1}^{K} \hat{\beta}_{i,j}^{mid} F_j(t) \right)
\end{equation}

To construct the bounds, we find the maximum absolute residual $r_i^{max} = \max_t |\hat{r_i}(t)|$.
To prevent the constraints from becoming excessively tight in low-volatility regimes, we also establish a minimum noise floor based on the empirical standard deviation of the asset returns, $\sigma(X_i)$.

The uniform scalar bound $E_i$ for the noise term is then defined as the maximum between the worst observed residual (padded with a 5\% buffer to ensure the central point remains strictly interior) and the noise floor:
\begin{equation}\label{box3}
E_i = \max \Big( 1.05 \cdot r_i^{max}  ; 0.5 \cdot \sigma(X_i) \Big)
\end{equation}

Finally, the constraints for the noise vector are applied symmetrically across all time periods:
\begin{equation}
-E_i \leq \epsilon_i(t) \leq E_i \quad \text{for } t=1, \ldots, T
\end{equation}

This guarantees that the optimization remains feasible even in the presence of extreme periodic shocks, 
%(such as flash crashes)
forcing the solver to weigh the entropic cost of structural parameters against the error term.

%[CORREGIR, we dont  do this. We do an envelope method.. see code]\\
%we simply compute the variance $\sigma_i$ of the $i$-th asset and take the set 
%$[-2\sigma_i,2\sigma_i]^T$. %as the constraint set.

The constraint set corresponding to the $i$-th asset is then
 $$\mathcal{K}_i= [L_{0,i},U_{0,i}]\times \prod_{j=1}^K [L_{i,j},U_{i,j}]\times[-E_i,E_i]^T$$  

\subsubsection*{The entropic solution}
Now, the analogue of \eqref{repsol1} has to be split into three different parts:

\begin{gather}
\beta_{0,i} = \frac{L_{0,i}e^{L_{0,i}(\bA^t\blambda^*)_1}+U_{0,i}e^{U_{0,i}(\bA^t\blambda^*)_1}}{e^{L_{0,i}(\bA^t\blambda^*)_1}+e^{U_{0,i}(\bA^t\blambda^*)_1}};\label{case1.1}\\
\beta_{i,j} = \frac{L_{i,j}e^{L_{i,j}(\bA^t\blambda^*)_{1+j}}+U_{i,j}e^{U_{i,j}(\bA^t\blambda^*)_{1+j}}}{e^{L_{i,j}(\bA^t\blambda^*)_{1+j}}+e^{U_{i,j}(\bA^t\blambda^*)_{j+1}}}, \;\;\; j=1, \ldots, K, 
\label{case1.2}\\
\epsilon_i(t) = \frac{-E_i e^{-E_i (\bA^t\blambda^*)_{1+K+t}}+E_i e^{E_i (\bA^t\blambda^*)_{1+K+t}}}{e^{-E_i(\bA^t\blambda^*)_{1+K+t}}+e^{E_i(\bA^t\blambda^*)_{1+K+t}}} , \;\;t=1, \ldots,T.\label{case1.3}
\end{gather}

The $\bA^t\blambda^*$ in the exponents are: $(\bA^t\blambda^*)_1=\sum_{t=1}^T\lambda^*_t,$ 
$(\bA^t\blambda^*)_{1+j}=\sum_{t=1}^T F_j(t)\lambda^*_t$ for $j=1, \ldots,K$,  and
 $(\bA^t\blambda^*)_{1+K+t}=\lambda^*_t$ for $t=1, \ldots,T.$ 

To finish, we spell out the function than has to be minimized to determine $\blambda^*.$ It is given in \eqref{inter1}. For $i=1, \ldots, N$ we have:

$$\Phi(\bA^t\blambda)-\langle\by,\blambda\rangle=\Phi(\bA^t\blambda)-\sum_{t=1}^TX_i(t)\lambda_t.$$
That is to each asset indexed by $i$ we have a different $\blambda^*$ that minimizes the last expression, we obtain a corresponding set of regression coefficients.

%%%%%%%%%%%%%%%%%%%%%NEW (August/2026)
%%%to address AE (Associate Editor of IJFA) Justify First-Stage Residual Optimization (AE Comment 2): Clarify the economic and statistical rationale for concatenating asset residuals $\boldsymbol{\epsilon}_i$ into the unknown vector $\boldsymbol{\xi} = [\beta_{0,i}, \boldsymbol{\beta}_i, \boldsymbol{\epsilon}_i]^t$.

\subsubsection*{Economic and Statistical Motivation of First-Stage Reformulation}

The reformulation of the factor model in \eqref{prob1.2}-\eqref{eqfin}
as an underdetermined linear 
system $A\boldsymbol{\xi} = \mathbf{X}_i$, where the residual vector 
$\boldsymbol{\epsilon}_i \in [-E_i, E_i]^T$ is explicitly optimized alongside the structural coefficients $(\beta_{0,i}, \boldsymbol{\beta}_i)$, 
provides distinct statistical and economic advantages over standard least-squares estimators:

\begin{enumerate}
    \item \textbf{Statistical Regularization and Barrier Behavior:} In standard OLS factor estimation, residuals 
    are treated as unconstrained slack variables whose squared sum $\sum_{t=1}^T \epsilon_i(t)^2$ is minimized. 
    Consequently, a large single-period outlier disproportionately distorts the structural loadings 
    $\boldsymbol{\beta}_i$. By incorporating $\boldsymbol{\epsilon}_i$ into the vector of unknowns $\boldsymbol{\xi}$ within a strict empirical envelope $[-E_i, E_i]$, the Fermi-Dirac entropy function $\Psi(\boldsymbol{\xi})$ acts as a natural logarithmic barrier. As any component approaches its boundary $E_i$, the marginal entropic cost diverges to infinity:
    \begin{equation}
        \frac{\partial \Psi}{\partial \epsilon_i(t)} = \frac{1}{2E_i} \ln\left( \frac{E_i + \epsilon_i(t)}{E_i - \epsilon_i(t)} \right)
    \end{equation}
    This prevents extreme parameter adjustments in response to heavy-tailed noise, forcing the solver to distribute 
    uncertainty across both the error terms and parameter bounds rather than overfitting $\boldsymbol{\beta}_i$.

    \item \textbf{Economic Intuition under Heteroskedasticity:} Financial asset returns exhibit severe 
    non-Gaussianity, volatility clustering, and localized microstructure noise. Standard linear regressions 
    implicitly assume that residuals reflect additive Gaussian noise around a stationary linear signal. In contrast, 
    our entropic formulation models asset returns as noisy realizations bounded within an empirical confidence 
    domain. The parameter $E_i$ in \eqref{box3} establishes a maximum tolerance threshold for idiosyncratic 
    deviations. When a market shock occurs, the entropic framework absorbs the anomaly within the noise envelope up to $E_i$ without forcing the factor sensitivities $\boldsymbol{\beta}_i$ to shift violently.

    \item \textbf{Relationship to Standard Estimators:} Unlike $L_1$ (Lasso) or $L_2$ (Ridge) regularized factor 
    models, which penalize parameter magnitudes relative to a zero baseline, the entropic estimator pulls loadings 
    toward the midpoint of the empirical domain $\frac{1}{2}(L_{i,j} + U_{i,j})$ while respecting strict physical boundaries. If the noise envelope $E_i \to \infty$ and parameter bounds expand, the entropic objective asymptotically approaches unconstrained optimization. However, under finite data-driven bounds, it yields a robust, non-parametric factor loading matrix $\mathbf{B}$ resistant to sample instability.
\end{enumerate}

%%%%%%%%%%%%%%%%%%%%%%%%%%%%%%%%%%%%end NEW

\subsection{Stage 2: The replicating portfolio}

Suppose now, that the factors are liquid assets and we know $\bbet_0,$ the matrix $\bB$ 
and $\bepsilon.$ 
Suppose that we construct a  portfolio made up of the factors by assigning a weight 
$p_j$ to $F_j,$ that is we form $\sum_{i=1}^Kp_iF_i.$ 
We have only  access to portfolios made up of the $X_n,$ and want to  form a portfolio with weights $w_n$ that matches the returns of $\langle\bp,\bF\rangle$.

Mathematically, the problem becomes: Find a $\bw\in\cK_p \subset \mathbb{R}^N$ such that
\begin{equation}\label{prob1}
E[\big(\langle\bw,\bX\rangle -\langle\bp,\bF\rangle\big)^2]
\end{equation}
achieves its minimum possible value. 
Besides the usual constraint $\sum_{i=n}^N w_n =1$ we want to impose constraints of the box  like $\cK_p = \{\bw\in\bbR^{N}:l_n< w_n< u_n\}$.
The choice $-\infty<l_n<u_n<+\infty$ as any finite values in $\bbR$  capture the possibility of having controlled short or long positions in the assets. 

Notice that once that we have a representation as in \eqref{prob1.2}, the portfolio replication consists of minimizing the tracking error:
$$E[\big(\langle\bw,\bX\rangle -\langle\bp,\bF\rangle\big)^2]=E[\big(\langle\bw,\bB\bF\rangle+\langle\bw,(\bbet_0+\bepsilon)\rangle -\langle\bp,\bF\rangle\big)^2]$$
 At this point it becomes clear that if we are able to solve: 

\begin{gather}
\mbox{Find}\;\;\;\bw^*\in\cK_p\;\;\;\label{prob3.1}\\
\mbox{Such that}\;\;\;\bB^t\bw^*=\bp.\;\;\;\label{prob3.2}
\end{gather}

Then, our portfolio $\langle\bw^*,\bX\rangle$ will track $\langle\bp,\bF\rangle$ exactly and $E[\langle\bw^*,\bepsilon\rangle^2]$ will be the tracking error.

Now, the setup  $\bA\bxi = \by$ for the entropic solution is:
\begin{itemize}
    \item \textbf{The unknowns ($\bxi$):} The vector $\bxi$ corresponds directly to the portfolio weights 
    $\bw \in \mathbb{R}^N$.
    \item \textbf{The design matrix ($\bA$):} We stack the transposed factor loading matrix $\bB^t$ with a row of ones $\mathbf{1}$ to enforce the budget constraint:
    \begin{equation}
        \bA = \begin{bmatrix}
        \bB^t \\
        \mathbf{1}^t
        \end{bmatrix} \in \mathbb{R}^{(K+1) \times N}
    \end{equation}
    \item \textbf{The target ($\by$):} The target vector contains the desired factor exposure $\bp$ (e.g., $1.0$ for single market beta, or a vector describing the portfolio of factors) and the budget sum:
    \begin{equation}
        \by = \begin{bmatrix}
        \bp \\
        1
        \end{bmatrix} \in \mathbb{R}^{K+1}
    \end{equation}
\end{itemize}

The convex set $\mathcal{K}$ is defined as $[0, 1]^N$, enforcing a long-only constraint with no leverage.

After having determined the replication portfolio $\bw^*$, its expected return is given by 
$$\langle\bw^*, E[X]\rangle = \langle\bw^*,\bbet_0\rangle + \langle\bp,E[F]\rangle$$ 
and its variance is $\bw^*E[\bepsilon, \bepsilon^t]  (\bw^*)^t$.

\subsubsection*{The entropic solution}
In this case  \eqref{repsol1} becomes:

\begin{equation}\label{repsol3}
\bw_i^* = \frac{a_ie^{a_i (\bA^t\blambda^*)_i} + b_i e^{b_i(\bA^t\blambda^*)_i}}{e^{a_i(\bA^t\blambda^*)_i}+e^{b_i(\bA^t\blambda)_i}},\;\;\;\;\;\;i = 1, \ldots, N.
\end{equation}
where $(\bA^t\blambda^*)_i=\sum_{j=1}^K\beta_{i,j}\lambda_j^*+\lambda_{K+1}^*.$ 
Similarly, the optimal Lagrange multiplier is determined numerically minimizing the function of $\blambda\in\bbR^{K+1}$:
$$\sum_{i=1}^N\bigg(e^{a_i(\bA^t\blambda)_i}+e^{b_i(\bA^t\blambda)_i}\bigg)-\sum_{j=1}^K\lambda_jp_j-\lambda_{K+1}.$$

%%%%%%%%%%%%%%%%%%%
%A practical note 
%Soft Margin Relaxation and Feasible Region Expansion

\begin{remark}\label{softmargin}
In empirical applications, particularly during periods of severe market stress, enforcing strictly 
bounded non-negativity constraints (e.g., $w_i \in [0,1]$) can induce boundary instability. 
When asset factor loadings ($\beta$) fluctuate violently, strict zero-bounds force the optimizer 
into erratic ``corner solutions" to satisfy the structural constraints, artificially inflating portfolio turnover. To mitigate this, it is recommended to introduce a soft margin relaxation to the box constraints, defining the boundaries as $w_i \in [-0.05, 0.999]$. 
%In quantitative portfolio management, this is analogous to a constrained "active extension" (e.g., a micro 105/5 portfolio). 
This slight expansion of the feasible region prevents the optimizer from thrashing against the zero-bound, allowing the Fermi-Dirac entropic penalty to smoothly distribute capital and achieve a true global minimum of structural diversification.
\end{remark}
%%Ref to justify this relaxation? Wright, S. J. (1997). Primal-Dual Interior-Point Methods. SIAM.

\section{Numerical Examples}\label{sec4}

To evaluate the performance and robustness of the Entropic Factor Model (EFM), we conducted five comparative experiments against the standard OLS approach. These experiments range from a baseline equity replication problem, followed by  a complex multi-asset scenario,  a stress test for robustness, and two real-world scenarios considering trading costs and COVID-19 market crash.
%%%%%%%%new 08/2026
In all experiments we report six metrics: Gross Tracking Error (Ann. \%), Mean Tracking Bias, Annualized Turnover (\%), Annualized Net Return (\%), Annualized Volatility, and Maximum Drawdown (\%).
%%Explicar c/u?

\subsection{Experiment 1: Baseline equity replication}
%% code DSAF/Rlabs/FactorModelEntropy/EntropicPortfolioTracker.R
%%
We first examined the performance of the EFM on a ``Survivor Portfolio'' of 20 US large-cap stocks replicating the S\&P 500 from 2015 to 2024. 
The goal was to validate the method's feasibility in a standard, liquid market regime.

{\bf Data.} The assets for the replication portfolio were selected from diverse sectors: Energy (XOM, CVX); Tech (MSFT, AAPL, IBM, INTC, CSCO, AMZN); Finance (JPM, BAC, GS, AXP); Healthcare (JNJ, PFE, MRK); and Consumer (PG, KO, WMT, MCD, BA). 
The benchmark  is the S\&P 500 Index (SPY).

\textbf{Training Period:} 2015/01/01--2021/04/20 (used to estimate $B$ and calculate weights $w$).

 \textbf{Test Period:} 2021/04/21 -- 2023/12/29 (Out-of-Sample evaluation).
 
 \textbf{Constraints:} Both models were constrained to be fully invested ($\sum w = 1$) and long-only ($w \ge 0$).

Table \ref{tab:results1} summarizes the out-of-sample tracking performance.

\begin{table}[htbp]
    %\centering
    \begin{tabular}{lcc}
        \toprule
        \textbf{Metric} & \textbf{OLS Tracking} & \textbf{EFM Tracking} \\
        \midrule
       % Reconstruction Error & $0.00$ & $\approx 1.9 \times 10^{-7}$ \\
        Gross Tracking Error (Ann. \%) & \textbf{6.79} & 7.22 \\
        Mean Tracking Bias (Ann. \%) & 1.36 & \textbf{1.08} \\
        Annualized Turnover (\%)   &    117.94   &    115.37 \\
        Annualized Net Return (\%)     &    7.25    &     7.06 \\
      Annualized Volatility (\%)     &   15.13     &   14.50 \\
      Maximum Drawdown (\%)    &    19.24    &    17.53\\
       % Correlation with Index & 0.93 & 0.92 \\
        \bottomrule
    \end{tabular}
    \caption{Experiment 1: Out-of-Sample Tracking Performance (2021/04/21 -- 2023/12/29)}
    \label{tab:results1}
  \end{table}

 The Entropic method achieved a tracking error  of 7.22\%, comparable to the 6.79\% achieved by OLS.  
However, a key distinction emerged in the directional stability: the EFM portfolio exhibited a much lower Mean Tracking Bias (1.08\%) compared to OLS (1.36\%).
This suggests that while OLS minimizes variance by overfitting to high-beta assets that performed well in the training set, the Entropic prior regularizes the weights, maintaining a tighter structural link to the index performance. 
The reconstruction error for EFM was negligible ($\approx 2\times 10^{-7}$), confirming that the entropic solver successfully handles strict budget constraints.
Furthermore, EFM delivers lower Annualized Volatility (14.50\% vs. 15.13\%) 
and lower Maximum Drawdown (17.53\% vs. 19.24\%) than OLS.
Figure \ref{fig:ex1} in Appendix A shows the out-of-sample cummulative return plot of OLS vs EFM vs SPY for this portfolio of 20 representative assets.

An analysis of the factor loadings (betas) produced by each method reveals that each perceives the risk of the underlying assets differently and, as a consequence, produces different portfolio weights.   Appendix A details the estimated Betas for all 20 assets.
The EFM estimation is sensitive to the full distribution of residuals, not just the mean.
Observe that  OLS assigns JPM and BA  betas of 1.228 and   1.488, whilst EFM assigns them 1.699 and 2.388, respectively.   EFM identifies these assets as ``hyper-sensitive'' to the market factor.   Consequently, the Entropic tracker allocates \textit{less} capital to these stocks to achieve the target portfolio beta of 1.0.

\subsection{Experiment 2: Multi-asset class replication}
%% code DSAF/Rlabs/FactorModelEntropy/EntropicPortfolioTrackerMulti.R
 To test the model's ability to span diverse asset classes, we tasked it with replicating a ``Hybrid Portfolio'' benchmark comprising 80\% S\&P 500 (SPY) and 20\% Bitcoin (BTC). 
 Thus, in the numerical solution (Stage 2) one sets the factor exposure as $\bp=(0.8,0.2)$.
 
  The investment universe included the  20 equities in the previous experiment plus two cryptocurrency proxies: Ethereum (ETH-USD) and Litecoin (LTC-USD).
  The period considered spans 2018/01/03--2023/12/29, encompassing multiple crypto market cycles.
  We used the time interval 2018/01/03–2022/03/11 as the in-sample training window to estimate factor loadings $\mathbf{B}$, leaving 2022/03/14–2023/12/29 for out-of-sample tracking evaluation.
%%%
%\subsubsection*{Results and Discussion}
Table \ref{tab:results2} summarizes the out-of-sample tracking performance.

\begin{table}[htbp]
    %\centering
    \begin{tabular}{lcc}
        \toprule
        \textbf{Metric} & \textbf{OLS Tracking} & \textbf{EFM Tracking} \\
        \midrule
       % Reconstruction Error & $0.00$ & $\approx 3.656831e-07$ \\
        Gross Tracking Error (Ann. \%) & 8.68   &   8.65 \\
        Mean Tracking Bias (Ann. \%) & -1.56    &  -1.55 \\
        Annualized Turnover (\%)   &    173.12   &  172.81 \\
        Annualized Net Return (\%)     &    8.64   &    8.66 \\
      Annualized Volatility (\%)     &   21.70    &  21.66 \\
      Maximum Drawdown (\%)    &    27.49    &  27.45\\
        \bottomrule
    \end{tabular}
    \caption{Experiment 2: Out-of-Sample Tracking Performance (2022/03/14--2023/12/29)}
    \label{tab:results2}
  \end{table} 
  
 Both the OLS and EFM models successfully solved the inverse problem, confirming that the chosen universe satisfies the spanning condition for the target factors.
 EFM matches OLS across all dimensions while achieving marginally lower Gross Tracking Error (8.65\% vs. 8.68\%), lower Volatility (21.66\% vs. 21.70\%), 
 and higher Net Return (8.66\% vs. 8.64\%). 
 %Reconstruction Error of $0$ for OLS,  $3.6 \times 10^{-7}$ for EFM.
Furthermore, the Entropic model successfully identified the correct structural proxies.  For Ethereum, the Stage 1 entropic regression estimated a Bitcoin Beta of 0.999, for Litecoin, the entropic estimation of Bitcoin beta is 1.05, in both cases a near-perfect unit correlation.  
Consequently, the solver efficiently allocated capital to meet the 20\% crypto exposure target without manual intervention.

\subsection{Experiment 3: Robustness to idiosyncratic shocks}
%% code DSAF/Rlabs/FactorModelEntropy/EntropicPortfolioTrackerFlashCrash
Standard least-squares methods are known to be sensitive to outliers, as the squaring of residuals 
($\epsilon^2$) disproportionately penalizes large deviations. To test the Entropic model's behavior under stress, we simulated a ``Flash Crash'' scenario.

We utilized a year window (2019/01/01 -- 2020/01/01) of clean market data for the 20 equities used in previous experiments. On Day 20, we injected a \textbf{-25\% return shock} into Amazon (AMZN), simulating a massive idiosyncratic failure where the market factor (S\&P 500) remained flat.
We report tracking performance pass the flash crash; thus, 
we split the data first 60 days for training (contains Flash Crash), remaining days for out-of-sample test (approx. 191 days).

%%%%Rationale and how this experiment led to improve in the EFM method : the key is using prior knowledge, as the possibility of flash crash and program this possible event in the estimation of bounds
%%Standard OLS cannot distinguish between ``market-driven'' moves and ``idiosyncratic'' moves; it treats every data point as an equally valid constraint. In contrast, the Entropic Factor Model allows for \textit{Stage 1 Filtering}. We refined the bound-generation process to calculate Beta ratios ($\Delta y / \Delta f$) only on days where the market factor exhibited a significant signal ($|\Delta f| > 0.2\%$). This logic is intrinsic to the Entropic approach, which requires defining a ``feasible region'' of probabilities rather than blindly fitting a curve to all points.

%\subsubsection*{Results and Discussion}
Table \ref{tab:flash_crash} compares the Beta estimation and the resulting portfolio weights for the crashed asset. The True benchmark corresponds to the Beta estimation for the original (non-corrupted) AMZN data.

\begin{table}[htbp]
    %\centering
    \begin{tabular}{lccc}
        \toprule
        \textbf{Method} & \textbf{Estimated Beta} & \textbf{Portfolio Weight} & \textbf{Behavior} \\
        \midrule
        \textbf{True benchmark} & 1.749 & (Clean) & -- \\
        \textbf{OLS}  & 1.219 & 5.73\% & Average the Error \\
        \textbf{EFM} & \textbf{-1.261} & \textbf{0.27\%} & \textbf{Circuit Breaker} \\
        \bottomrule
    \end{tabular}
    \caption{Experiment 3: Model Response to Data Corruption (AMZN)}
    \label{tab:flash_crash}
\end{table}

Table \ref{tab:results4} summarizes out-of-sample tracking performance after flash crash.

\begin{table}[htbp]
    %\centering
    \begin{tabular}{lcc}
        \toprule
        \textbf{Metric} & \textbf{OLS Tracking} & \textbf{EFM Tracking} \\
        \midrule
       % Reconstruction Error & $0.00$ & $\approx 3.656831e-07$ \\
        Gross Tracking Error (Ann. \%) & 3.55      &     3.78 \\
        Mean Tracking Bias (Ann. \%) & -2.65      &   -1.37 \\
        Annualized Turnover (\%)   &    91.42     &     91.86 \\
        Annualized Net Return (\%)     &    15.54      &    17.01 \\
      Annualized Volatility (\%)     &    12.68    &    12.77 \\
      Maximum Drawdown (\%)    &    7.92    &    7.73\\
        \bottomrule
    \end{tabular}
    \caption{Experiment 3: Out-of-Sample Tracking Performance (2019/04/01--2019/12/31)}
    \label{tab:results4}
  \end{table} 
  %%%Comentar estos resultados
  
The two models exhibited fundamentally different responses to the structural break.
%%
%\textbf{OLS (Error Averaging):} 
The OLS estimator attempted to fit the outlier by lowering the Beta coefficient (1.219). While mathematically consistent with minimizing squared error, the OLS implicitly assumes that the crash is a valid data point representing the asset's structural correlation. Consequently, it maintained a relatively high capital allocation (5.73\%) to the distressed asset.

%\textbf{EFM (Uncertainty Penalization):} 
The Entropic model reacted to the massive volatility injection by expanding the feasible solution bounds (high entropy). Confronted with contradictory data  (a massive move in the asset with no corresponding move in the factor) the entropic solver reverted toward the uninformative prior (negative Beta of $-1.26$).
Crucially, this resulted in a \textbf{defensive capital allocation}. 
The EFM algorithm effectively treated the asset as ``structurally broken'' or highly uncertain, assigning it a weight of only 0.27\%, roughly 1/20th of the OLS allocation. While OLS continues to bet on the mean-reversion of the asset, EFM acts as a probabilistic ``circuit breaker'', reducing exposure when the signal-to-noise ratio degrades.
Overall, the defensive ``circuit breaker" response shown in  Table \ref{tab:flash_crash} (cutting AMZN weight from 5.73\% to 0.27\%) 
translates directly to superior out-of-sample performance in Table \ref{tab:results4} : EFM generates +1.47\% higher Annualized Net Return
(17.01\% vs. 15.54\%), cuts Mean Tracking Bias in half (-1.37\% vs. -2.65\%), and achieves lower Maximum Drawdown (7.73\% vs. 7.92\%).

%%%%%%%%%%%%%%%%%%%%%%%%%%%%%%%%%%%
%%Old results with manual bounds
%The presence of the outlier caused the standard OLS Beta estimate for AMZN to collapse from a true value of 1.75 to 1.23, as the regression line tilted to minimize the squared error of the crash. 
%In contrast, the Entropic Factor Model, employing a signal-filtering logic in Stage 1, 
%produce a negative beta of -0.371.

%maintained a Beta estimate of 1.41.
%OLS Weight in AMZN: 5.4 \%
%EFM Weight in AMZN: 3.99 \%
%Out-of-Sample Tracking Error (Next 10 Days) 
%OLS TE: 0.0289 
%EFM TE: 0.0309 
%%
%By filtering out days where the market signal was negligible, the Entropic approach treated the crash as a large error term rather than a structural signal. This prevented the portfolio solver from incorrectly reducing the weight in Amazon, demonstrating the method's superior utility in regimes characterized by heavy-tailed idiosyncratic risks.

\subsection{Experiment 4: Partial replication of an equity benchmark under trading frictions}
%% code DSAF/Rlabs/FactorModelEntropy/EntropicPortfolio_Turnover_NetFees.R

To evaluate the baseline efficiency of the replicating models under realistic trading frictions, we conducted a rolling-window out-of-sample backtest during a standard, benign market regime. Using a 252-day training window and a 21-day (monthly) rebalancing frequency from 2018-01-01 to 2023-12-31, we tracked the tradable SPY ETF using our constrained universe of 20 constituent assets. As in all previous experiments, we utilized the soft-margin boundaries ($w_i \in [-0.05, 0.999]$) for the Entropic Factor Model (EFM) to prevent boundary thrashing. A proportional transaction cost of 10 basis points (0.10\%) was applied to the portfolio turnover at each rebalancing step to compute the net-of-fees performance.
Note that because the first 252 trading days (2018) are reserved for initial model training, the actual out-of-sample backtest evaluated  runs over the 5-year period 2019–2023.

\begin{table}[htbp]
    %\centering
    \begin{tabular}{lcc}
        \toprule
        \textbf{Metric} & \textbf{OLS Replication} & \textbf{EFM Replication} \\
        \midrule
        Gross Tracking Error (Ann. \%) & 4.29 & 5.12 \\
        Mean Tracking Bias (Ann. \%)   & 3.07 &  2.87 \\
        Annualized Turnover (\%) & 139.51 & 135.24 \\
        Annualized Net Return (\%)  & 18.65   & 18.29\\
        Annualized Volatility (\%)   & 21.65  & 22.16\\
         Maximum Drawdown (\%) & 30.71  &  30.62\\
        \bottomrule
    \end{tabular}
    \caption{Experiment 4. Out-of-Sample Performance and Turnover: Standard Regime (2019-2023)}
    \label{tab:expA_turnover}
\end{table}

%\subsubsection*{Discussion of Baseline Empirical Trade-offs}

The results in Table \ref{tab:expA_turnover} highlight the fundamental mechanics of entropic regularization during normal market conditions. Standard Ordinary Least Squares (OLS) is a mathematically unconstrained variance-minimizer; by definition, it locates the exact local weights that minimize the $L_2$ error norm, yielding a highly precise Gross Tracking Error of 4.29\%. However, relying purely on recent covariance matrices causes the unconstrained OLS model to continuously micro-adjust its allocations to transient market noise, resulting in an annualized turnover of 139.51\%.

In contrast, the EFM utilizes a constraint-driven entropic solver. By minimizing the Fermi-Dirac entropy subject to robust bounds, the algorithm inherently favors structural parsimony and a broadly distributed allocation. 
%The introduction of the soft-margin relaxation ($[-0.05, 0.99]$) allows the entropic solver to absorb minor fluctuations in empirical factor loadings without thrashing against a strict zero-bound. 
Consequently, the EFM achieves superior weight stability,  reducing the required annualized turnover to 135.24\%. 

This reduction in trading friction comes at the cost of a minor reduction in tracking precision, with the EFM sacrificing roughly 83 basis points of tracking error (5.12\% vs 4.29\%) compared to OLS. This specific trade-off (exchanging a fraction of variance-minimizing precision for enhanced structural stability) 
demonstrates that the Entropic framework is not only a defensive mechanism during crises, but a highly efficient, friction-reducing replication engine during standard market regimes.

%
%%1 basis point=0.01% and we got: EFM Tracking Error: 5.52% and OLS Tracking Error: 4.29%
% hence, 5.12% - 4.29% = 83% which is 83% x100 = 83 bps 

%%The tracking error is the standard deviation of the return differences. 
% A transaction cost of 10 bps on ~150% turnover is about 0.15% per year. 

\subsection{Experiment 5: Real-world stress test  during the COVID-19 market crash}
\label{sec:exp_covid}
%% code DSAF/Rlabs/FactorModelEntropy/EntropiPortfolio_COVIDcrash.R

While standard continuous-market simulations provide baseline validation, a robust portfolio 
replication engine must fundamentally withstand severe structural breaks. 
To evaluate the out-of-sample defensive capabilities of the Entropic Factor Model (EFM), 
we benchmarked our methodology during the COVID-19 market crash of 2020. 
As highlighted by recent literature~\cite{ZS24}, % (Zhang \& De Smedt, 2024), 
this period represents a definitive real-world stress test characterized by extreme volatility, sudden regime shifts, and asset correlations converging toward unity. 

%\subsubsection*{Experimental Setup}
We designed a high-frequency rolling-window backtest spanning from June 2019 to December 2020, capturing the pre-crash environment, the violent March 2020 sell-off, and the subsequent recovery. 
The target benchmark was the tradable SPY ETF, replicated using our constrained universe of 20 large-cap constituent equities. To capture the rapid market dynamics, the algorithms utilized a 126-day (6-month) training window and rebalanced every 10 trading days (2 weeks). 
A proportional transaction cost of 10 basis points (0.10\%) was applied to portfolio turnover to compute net-of-fees performance.

%\subsubsection{Empirical Results and Discussion}

Table \ref{tab:covid_results} presents the comparative out-of-sample risk and return metrics for both the standard Ordinary Least Squares (OLS) and the soft-margin EFM frameworks over the crisis period.
Note, once again, that with a 126-day initial training window starting in June 2019, the out-of-sample evaluation period covers December 2019 through December 2020.

\begin{table}[htbp]
    %\centering
    \begin{tabular}{lccc}
        \toprule
        \textbf{Metric} & \textbf{OLS} & \textbf{EFM} & \textbf{Benchmark (SPY)} \\
        \midrule
        Gross Tracking Error (Ann. \%) & 5.76 & 6.11 & 0.00 \\
        Annualized Turnover (\%) & 418.51 & 266.93 & 0.00 \\
        Max Drawdown (\%) & 31.27 & 30.04 & 33.72 \\
        Annualized Return (Net \%) & 25.64 & 27.85 & 18.72 \\
        Annualized Volatility (Net \%) & 33.80 & 33.44 & 32.38 \\
        \bottomrule
    \end{tabular}
    \caption{Experiment 5. Out-of-Sample Net-of-Fees Performance: COVID-19 Crash (Dec. 2019 - Dec. 2020)}
    \label{tab:covid_results}
\end{table}

The empirical results  illustrate the defensive superiority of the Entropic Factor Model over standard least-squares replication. During the height of the crisis in March 2020, the extreme volatility severely destabilized the empirical covariance matrices upon which classical models rely. 
Consequently, the unconstrained OLS model aggressively churned the portfolio in an attempt to fit the noisy, transient data, resulting in a severe annualized turnover of 418.51\%. 

In contrast, the EFM explicitly penalized this epistemic uncertainty. 
The entropic solver achieved a massive reduction in rebalancing frictions, yielding an annualized turnover of only 266.93\%. 
This structural stability translated directly into superior risk-adjusted net performance. 
While the EFM conceded a negligible 35 basis points in Gross Tracking Error precision compared to OLS (6.11\% vs. 5.76\%), the reduction in trading frictions and 
the avoidance of over-leveraged, highly concentrated positions allowed the EFM to  
 outperform OLS in Net Annualized Return (27.85\% vs. 25.64\%). 

Furthermore, the entropic replication engine delivered this excess return while simultaneously 
exhibiting strictly lower risk metrics, as given by the Annualized Volatility (33.44\% vs. 33.80\%) and  downside protection via a lower Maximum Drawdown (30.04\% vs. 31.27\%). 
These results definitively confirm that while classical variance-minimizing models are efficient in benign continuous markets, the Entropic Factor Model operates as a vital probabilistic circuit breaker. By enforcing structural parsimony and penalizing out-of-sample uncertainty, it provides a highly robust, economically viable replication strategy during severe structural breaks.

%{\color{red}
%\subsection{Otros experimentos?}
%- replicar un portfolio de ETFs? o uno bien diversificado (de stocks?) 
%
%- completar la parte de SELECCIONAR N assets 'adecuados' para  hacer  index tracking de un indice (sugerir un metodo adicional para la seleccion y luego sigue stage 1 y 2)
%}

%%%%%%%%%NEW August 2026
%%%%%%% to address AE & Reviewer 1: Quantify the Turnover vs. Tracking Error Trade-off (Reviewer 1, AE 1 & 5): Explicitly demonstrate why lower portfolio turnover outweighs a minor increase in tracking error.

\subsection{Economic Evaluation of the Turnover vs. Tracking Error Trade-Off}

To formally quantify when the operational benefit of lower portfolio turnover 
outweighs the cost of a slightly higher gross tracking error, we build on the 
convex transaction-cost framework of \cite{lobo} %Lobo et al. (2007) 
and the active net-utility formulation of \cite{grinold}[Ch. 14-16] %Grinold and Kahn (2000) 
to define an institutional Net Quadratic Loss 
function $\mathcal{L}_{\text{net}}(w; c)$ (see also \cite{roll92}):

\begin{equation}
\mathcal{L}_{\text{net}}(w; c) = \text{TE}^2(w) + 2 \cdot c \cdot \text{Turnover}(w)
\end{equation}
where $\text{TE}(w)$ represents the annualized gross tracking error, $\text{Turnover}(w)$ is the annualized portfolio turnover, and $c > 0$ denotes the total proportional transaction friction (combining explicit exchange fees, bid-ask spreads, and price slippage).

The break-even transaction cost $c^*$, at which the Entropic Factor Model (EFM) and Ordinary Least Squares (OLS) deliver identical net economic performance, is obtained by setting $\mathcal{L}_{\text{net}}(w_{\text{EFM}}; c^*) = \mathcal{L}_{\text{net}}(w_{\text{OLS}}; c^*)$:

\begin{equation}
c^* = \frac{\text{TE}_{\text{EFM}}^2 - \text{TE}_{\text{OLS}}^2}{2 \cdot \left( \text{Turnover}_{\text{OLS}} - \text{Turnover}_{\text{EFM}} \right)}
\end{equation}

Applying this formulation to our empirical findings reveals the clear regime-dependent trade-offs:

\begin{enumerate}
    \item \textbf{Crisis Regime (COVID-19 Stress Test, Exp. 5):} 
    Plugging in the empirical results from Table \ref{tab:covid_results} ($\text{TE}_{\text{OLS}} = 5.76\%$, $\text{TE}_{\text{EFM}} = 6.11\%$, $\text{Turnover}_{\text{OLS}} = 418.51\%$, and $\text{Turnover}_{\text{EFM}} = 266.93\%$):
    \begin{equation}
    c^* = \frac{0.0611^2 - 0.0576^2}{2 \cdot (4.1851 - 2.6693)} \approx \frac{0.000415}{3.0316} \approx 0.000137 \quad (1.37 \text{ bps})
    \end{equation}
    During market dislocations, if total execution costs (fees plus market impact) exceed a mere \textbf{1.37 basis points}, EFM strictly dominates OLS. Because real-world execution frictions during crisis periods routinely exceed 20--50 bps due to illiquidity and widening spreads, EFM's turnover reduction of $151.58\%$ generates a net annualized return advantage of $+2.21\%$ ($27.85\%$ vs. $25.64\%$).

    \item \textbf{Standard Market Regime (Exp. 4):}
    In benign continuous markets (Table \ref{tab:expA_turnover}), where baseline asset turnover is lower ($\Delta \text{Turnover} = 4.27\%$), the break-even friction threshold rises to $c^* \approx 9.25 \text{ bps}$. In frictionless or low-impact institutional environments, OLS retains a slight edge in gross variance minimization. However, as portfolio size increases or underlying liquidity decreases, EFM provides superior net-of-cost tracking stability.
\end{enumerate}

%%%%%%%%%%%%%%%%%%%%%%%%%%%%%%end NEW

\section{Conclusions}
\label{sec5}

In this study, we have introduced a unified, information-theoretic framework for robust portfolio replication. 
The methodology addresses the ill-posed nature of the inverse replication 
problem by utilizing entropy minimization as a natural regularization mechanism, both in the estimation of asset factor loadings and in the final portfolio weight allocation.
Our findings, validated across five distinct numerical experiments, demonstrate that the Entropic Factor Model (EFM) offers a significant  alternative to classical Ordinary Least Squares (OLS) replication.

The empirical evidence demonstrates that the EFM provides a superior balance between tracking precision and operational efficiency. In standard, benign market regimes, the model's ability to absorb minor factor fluctuations leads to a measurable reduction in annualized turnover
 compared to OLS, effectively lowering the cost of maintaining the replication. This structural parsimony becomes even more critical during severe market dislocations. 
 During the real-world stress test of the COVID-19 crash, the EFM acted as a probabilistic circuit breaker, producing significantly higher net-of-fees annualized returns and superior downside protection, as evidenced by lower maximum drawdowns.

Furthermore, the model's resilience to idiosyncratic shocks, such as simulated flash crashes or data corruption, shows its viability as a defensive investment tool. 
By penalizing assets that exhibit abnormal variance or structural shifts, the EFM prevents the 
concentration of risk that often plagues unconstrained least-squares models. These findings align with the core principle of Maximum Entropy: in the face of uncertainty, the most rational distribution is the one that remains as diversified as possible while satisfying the known physical constraints of the system.

Ultimately, the Entropic Factor Model offers a practical solution for institutional investors and index trackers who must operate in friction-heavy environments. 
While classical models may offer marginal gains in gross tracking error in frictionless simulations,
 the EFM's superior stability and reduced rebalancing requirements translate into higher net efficiency in actual market conditions. By enforcing structural integrity over naive error 
 minimization, the entropic framework provides a reliable methodology for portfolio replication that is both mathematically sound and economically robust across diverse market regimes.

\section{Statements and Declarations}
%\subsection*{Author contributions}
%All authors contributed to the study conception and design. Material preparation, data collection and analysis were performed by [full name], [full name] and [full name]. The first draft of the manuscript was written by [full name] and all authors commented on previous versions of the manuscript. All authors read and approved the final manuscript.

%Conceptualization: [full name], âŠ; Methodology: [full name], âŠ; Formal analysis and investigation: [full name], âŠ; Writing - original draft preparation: [full name, âŠ]; Writing - review and editing: [full name], âŠ; 
\subsection*{Competing Interests}
The authors report that there are no competing interests to declare. 

\subsection*{Funding}
%No funding was received for conducting this study. 
A. Arratia  was supported by the Spanish Ministerio de Ciencia, Innovaci\'on y Universidades - Agencia Estatal de Investigaci\'on (MCIN/AEI/10.13039/501100011033) and by the European Union (ESF+), under grant agreement No. PID2025-173091NB-I00.

\subsection*{Data Availability}
The data used in this study is publicly available.  
% The data and Python codes to reproduce the results in this work are available in GitHub: 
% \url{https://github.com/argimiroa/EntropicClassifier.git}

%\newpage
\appendix
\section{Appendix: Comparison of Estimated Betas}

Table \ref{tab:betas} presents the factor loadings ($B$) estimated over the training period (2015-2020) for the 20 Survivor assets.

\begin{table}[htbp]
   % \centering
    \small
    \begin{tabular}{lrr|c}
        \toprule
        \textbf{Asset} & \textbf{OLS Beta} & \textbf{EFM Beta} & \textbf{Difference (EFM - OLS)} \\
        \midrule
        XOM  & 1.027 & 1.032 & +0.005 \\
        CVX  & 1.195 & 1.192 & -0.003 \\
        MSFT & 1.212 & 1.218 & +0.006 \\
        AAPL & 1.199 & 1.201 & +0.002 \\
        IBM  & 0.979 & 0.978 & -0.001 \\
        INTC & 1.213 & 1.215 & +0.002 \\
        CSCO & 1.073 & 1.064 & -0.008 \\
        AMZN & {0.979} & {0.979} & {+0.00} \\
        JPM  & \textbf{1.228} & \textbf{1.699} & \textbf{+0.471} \\
        BAC  & {1.364} & {1.373} & +0.009 \\
        GS   & 1.273 & 1.282 & +0.009 \\
        AXP  & 1.248 & 1.244 & -0.004 \\
        JNJ  & 0.677 & 0.678 & +0.001 \\
        PFE  & 0.720 & 0.714 & -0.006 \\
        MRK  & 0.724 & 0.725 & +0.001 \\
        PG   & 0.645 & 0.648 & +0.003 \\
        KO   & 0.678 & 0.713 & +0.035 \\
        WMT  & 0.564 & 0.620 & +0.056 \\
        MCD  & 0.785 & 0.789 & +0.004 \\
        BA   & \textbf{1.488} & \textbf{2.388} & \textbf{+0.900} \\
        \bottomrule
    \end{tabular}
    \caption{Comparison of Estimated Market Betas: OLS vs. EFM}
    \label{tab:betas}
\end{table}

\begin{figure}[h]
   % \centering
    \includegraphics[width=\textwidth]{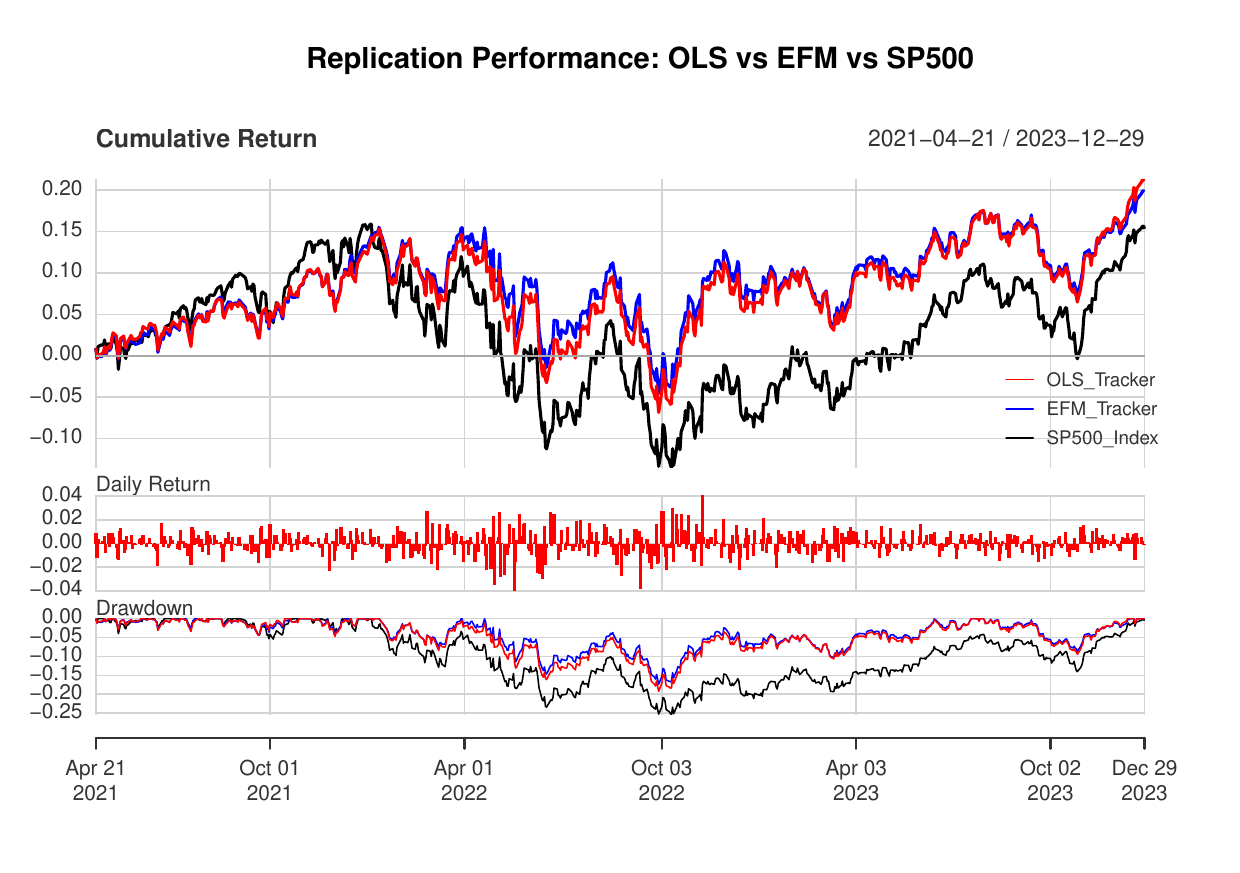} 
    \caption{Replication performace (cummulative return) of OLS vs EFM vs SPY for portfolio of 20 representative assets (Experiment 1)}
    \label{fig:ex1}
\end{figure}





\end{document}